\documentclass[aps,pra,reprint,showpacs,superscriptaddress,notitlepage,10pt]{revtex4-1}
\usepackage{amsmath}
\usepackage{amssymb}
\usepackage{amsthm}
\usepackage{bbm}
\usepackage{graphicx}
\usepackage{hyperref}
\usepackage{physics}
\usepackage{systeme}
\usepackage{times}

\RequirePackage{color}

\newcommand{\e}[1]{\{#1\}}
\newtheorem{theorem}{Theorem}
\newtheorem{corollary}{Corollary}[theorem]

\newtheorem{claim}{Claim}
\DeclareMathOperator{\diag}{diag}

\hypersetup{
  colorlinks=true,linkcolor=blue,citecolor=blue,
  filecolor=blue,urlcolor=blue,breaklinks=true
}

\begin{document}

\title{Unitary equivalence between the Green's function and Schr\"odinger
approaches for quantum graphs}

\author{Fabiano M. Andrade}
\email{fmandrade@uepg.br}
\affiliation{
  Departamento de Matem\'{a}tica e Estat\'{i}stica,
  Universidade Estadual de Ponta Grossa,
  84030-900 Ponta Grossa-PR, Brazil
}

\author{Simone Severini}
\affiliation{
  Department of Computer Science,
  University College London,
  WC1E 6BT London, United Kingdom
}
\affiliation{
  Institute of Natural Sciences,
  Shanghai Jiao Tong University,
  Shanghai, China
}

\date{\today}

\begin{abstract}
In a previous work [Andrade \textit{et al.}, Phys. Rep. \textbf{647}, 1
(2016)], it was shown that the exact Green's function (GF) for an
arbitrarily large (although finite) quantum graph is given as a sum over
scattering paths, where local quantum effects are taken into account
through the reflection and transmission scattering amplitudes.
To deal with general graphs, two simplifying procedures were developed:
regrouping of paths into families of paths and the separation of a
large graph into subgraphs.
However, for less symmetrical graphs with complicated topologies as, for
instance, random graphs, it can become cumbersome to choose the
subgraphs and the families of paths.
In this work, an even more general procedure to construct the energy
domain GF for a quantum graph based on its adjacency matrix is
presented.
This new construction allows us to obtain the secular determinant,
unraveling a unitary equivalence between the scattering Schr\"odinger
approach and the Green's function approach.
It also enables us to write a trace formula based on the Green's function
approach.
The present construction has the advantage that it can be applied
directly for any graph, going from regular to random topologies.
\end{abstract}


\maketitle

The past decade witnessed a notable interest in the interplay between
quantum mechanics and graphs.
The area is very rich because its objectives go from tests in spin
chains as nanodevices to the explanation of natural phenomena as energy
transfer in biological systems.
General methods to deal with graphs are always very welcome because the
myriad of different topologies make it difficult to develop a unique
method that holds for all graphs.
In the context of quantum graphs
\cite{JCP.4.673.1936,JCP.21.1565.1953,CRASPSIM.296.793.1983,
PRL.79.4794.1997,AoP.274.76.1999,PRL.84.1427.2000},
a Green's function (GF) approach was first proposed in
\cite{JPA.36.545.2003} and explored in depth in
\cite{PR.647.1.2016}.
In the latter, to handle general quantum graphs of different
topologies, two simplification procedures were developed:
(i) the regrouping of infinite many scattering paths into finite
families of paths (FP)
and
(ii) the division of the graph into subgraphs, then solving each
subgraph individually by calculating effective scattering amplitudes,
and then connecting all the pieces altogether.
As described in \cite{PR.647.1.2016}, the GF construction
based on these two procedures is very general and useful.
However, for large graphs, less symmetrical graphs, graphs that
change the connections by some mechanism, or random graphs, it may
become really difficult to choose the subgraphs and to define the
FP.
Furthermore, although the final result is totally independent of the
choices of the FP, this choice is not unique, preventing, for
instance, the development of a general algorithm for the GF
construction.

In this paper, we aim to give an even more general and powerful method
for the GF construction for quantum graphs.
We shall show that the GF approach (GFA) presented here provides an
alternative derivation for the secular determinant, unraveling a unitary
equivalence between the GFA and the scattering Schr\"odinger approach
(SSA)
\cite{PSPM.77.291.2008,Book.2012.Berkolaiko,arXiv:1603.07356}.
Moreover, it  also provides another way to derive a trace formula for
quantum graphs
\cite{CRASPSIM.296.793.1983,PRL.79.4794.1997,AoP.274.76.1999,
AHP.10.189.2009,CM.447.175.2007}.

A \textit{graph} $X(V,E)$ is defined as a pair consisting of a set of
vertices $V(X)=\{1,\ldots,n\}$  and a set of edges
$E(X)=\{e_{1},\ldots,e_{l}\}$, where each edge is a pair of vertices
\cite{Book.2010.Diestel}.
The graph topology is described in terms of the adjacency matrix $A(X)$
of dimension $n \times n$ where the  $ij$th element $A_{i j}(X)$ is $1$
if $\{i,j\}\in E(X)$ and zero  otherwise.
Two vertices are neighbors whether they are connected by an edge.
The set $E_i=\left\{j:\{i,j\}\in E(X)\right\}$ is the neighborhood of
the vertex $i\in V(X)$.
We denote by $E_{i}^{k}= E_{i}\setminus \e{k}$ the set of neighbors
of the vertex $i$, but  with the vertex $k$ excluded.
The degree of $i$ is $d_i=|E_i|=\sum_{j=1}^{n} A_{ij}(X)$.
These definitions refer to \textit{discrete} graphs.
To discuss quantum graphs, it is necessary to equip the graphs with a
metric.
A \textit{metric graph} $\Gamma(V,E)$ is a graph in which is assigned a
positive length $\ell_{e_{s}}\in(0,+\infty)$ to each edge, thus defining
the set $\boldsymbol{\ell}=\{\ell_{e_{1}},\ldots,\ell_{e_{l}}\}$.
When a single ended edge $e_{s}$ is taken as semi-infinite
($\ell_{e_{s}} = + \infty$), it is called a ``lead.''
A \textit{quantum graph} is a metric graph in which it is possible to
define a Schr\"odinger operator along with appropriated boundary
conditions (BCs) at the vertices, or more formally, a triple
$\{\Gamma(V,E),H,\bold{bc}\}$ with $H$ a differential operator and
$\bold{bc}$ a set of BCs.
For the free Schr\"{o}dinger operator $H=-(\hbar^2/2m)d^2/dx^2$ it leads
us to the eigenvalue equation
\begin{equation}
  \label{eq:eigen}
  - \psi_{\e{i,j}}''(x) = k^2 \psi_{\e{i,j}}(x),
\end{equation}
where $k=\sqrt{2mE/\hbar^2}$, $m$ is the mass, $E$ is the energy, and
$\psi_{\e{i,j}}$ is the wave function on the edge $\e{i,j}$.
Hereafter we consider just simple connected graphs.

An important ingredient in the GFA for quantum
graphs is the individual scattering amplitudes defined at each one of
the graph vertices, in a such way that we can define a scattering matrix
$\boldsymbol{\sigma}_{j}$ for each vertex $j$ of the graph.
The scattering amplitudes are entirely determined by the BCs defined at
each vertex and the most general ones, which are consistent with quantum
flux conservation and fulfill the required condition of
self-adjointness, were discussed in \cite{JPA.32.595.1999}.
Without loss of generality, in an arbitrary
graph locally we can always treat a vertex $j$ with its edges as a star
graph.
A star graph on $n$ vertices, $S_n$, is a graph where one
central vertex has degree $n-1$ and all others vertices have degree
$1$.
Consider thus a star graph as the one depicted in Fig. \ref{fig:fig1}
and let
$\Psi(j)=\left(\psi_{\{j,1\}}(j),\ldots,\psi_{\{j,n\}}(j)\right)^{T}$.
The most general BC that are consistent with the self-adjoint condition
\cite{PR.647.1.2016} are totally defined by two $d_{j} \times d_{j}$
matrices $\mathcal{A}_{j}$ and $\mathcal{B}_{j}$ such that
\cite{JPA.32.595.1999}
\begin{equation}
  \label{eq:bc}
  \mathcal{A}_{j} \Psi(j)+ \mathcal{B}_{j} \Psi'(j) = 0,
\end{equation}
the matrix $\mathcal{A}\mathcal{B}^{*}$ is self-adjoint, and the
$d_{j} \times 2d_{j}$ matrix $\left(\mathcal{A}_j,\mathcal{B}_j\right)$
has the maximal rank $d_{j}$.
The scattering amplitudes associated with the BC \eqref{eq:bc}, can be
determined by considering a plane wave on the edge $\e{i,j}$ incident on
the vertex $j$ with degree $d_{j}$.
Thus the scattering solutions that satisfy the eigenvalue equation
\eqref{eq:eigen} are given by
\begin{align}
  \label{eq:wavef}
  \psi_{\e{i,j}}(x) = {}
  & e^{- i k x}
    + \sigma_{j}^{[\e{j,i},\e{i,j}]}(k) e^{i k x},\nonumber \\
  \psi_{\e{j,l}}(x) = {}
  &  \sigma_{j}^{[\e{j,l},\e{i,j}]}(k) e^{i k x}.
\end{align}
The quantities
$\sigma_{j}^{[\e{j,i},\e{i,j}]}(k)=r_{j}^{[\e{j,i},\e{i,j}]}(k)$
and
$\sigma_{j}^{[\e{j,p},\e{i,j}]}(k)=t_{j}^{[\e{j,p},\e{i,j}]}(k)$
are the reflection and transmission amplitudes at the vertex $j$,
respectively.
By applying the BC \eqref{eq:bc}, we have
\begin{equation}
  \label{eq:sigma}
  \boldsymbol{\sigma}_{j} (k) =  -
  (\mathcal{A}_{j} + i k \mathcal{B}_{j})^{-1}
  (\mathcal{A}_{j} - i k \mathcal{B}_{j}).
\end{equation}
\begin{figure}
  \centering
  \includegraphics*[width=0.30\columnwidth]{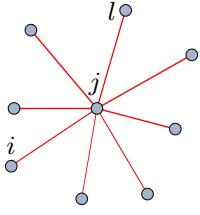}
  \caption{\label{fig:fig1}
    Locally, any graph looks like a star graph.}
\end{figure}
So, we can observe that the $\boldsymbol{\sigma}_{j}$ generally depends
on $k$ in a non-trivial manner. 
However, there are certain BCs that are independent of
$k$, as, for instance, the case of Dirichlet, Neumann, and Kirchoff
BCs \cite{AHP.10.189.2009}.
Thus we can see that for quantum graphs it is totally equivalent to set
either the BC or to specify the scattering matrix at the vertex $j$
\cite{JPA.32.595.1999}.
As said above, the $\boldsymbol{\sigma}_{j}(k)$ must satisfy the
requirement of quantum flux conservation, so it demands that the
$\boldsymbol{\sigma}_{j}(k)$ must be unitary,
$\boldsymbol{\sigma}_{j}(k)\boldsymbol{\sigma}_{j}^{\dagger}(k)
=\mathbbm{1}$,
and $\boldsymbol{\sigma}_{j}(k)=\boldsymbol{\sigma}_{j}^{\dagger}(-k)$,
leading to
\begin{align}
  \label{eq:sigma_rel}
  \sigma_{j}^{[\e{j,l},\e{i,j}]}(k) =
    \left[\sigma_{j}^{[\e{i,j},\e{j,l}]}(-k)\right]^{*},\nonumber \\
  \sum_{i \in E_{j}}
  \sigma_{j}^{[\e{j,l},\e{i,j}]}(k)
  \left[\sigma_{j}^{[\e{j,m},\e{i,j}]}(k)\right]^{*} =
  \delta_{lm}, \\
  \sum_{i \in E_{j}}
  \sigma_{j}^{[\e{i,j},\e{j,l}]}(k)
  \left[\sigma_{j}^{[\e{i,j},\e{j,m}]}(k)\right]^{*} =
    \delta_{lm}\nonumber,
\end{align}
which are natural generalizations of the usual relations for the
scattering amplitudes in 1D scattering problems
\cite{Book.1989.Chadan}.

Consider a quantum graph $\{\Gamma(V,E),H,\mathbf{bc}\}$ with the
adjcency matrix $A(\Gamma)$.
Then, add two leads $e_i$ and $e_f$ to the vertices $1$ and $n$,
respectively, as shown in Fig.  \ref{fig:fig2}, turning it into an
open quantum graph, suitable for studying scattering problems.
The exact scattering GF for a particle of fixed energy
$E=\hbar^2k^2/2m$, with initial position $x_i$ in the lead $e_i$ and
final position $x_f$ in the lead $e_{f}$, is given by a sum over all the
scattering paths (SP) connecting the points $x_i$ and $x_f$, where each
path is weighted by the product of the scattering amplitudes gained
along the path.
These scattering amplitudes are determined through the BCs defined at
the vertices.
Thus the exact scattering GF is written as
\cite{CM.415.201.2006,CM.447.175.2007} (see also
Ref. \cite{PR.647.1.2016})
\begin{equation}
  G(x_f,x_i;k) = \frac{m}{i\hbar^{2} k}
  \sum_{\rm SP} W_{\rm SP} e^{[\frac{i}{\hbar}S_{\rm SP}(x_f,x_i;k)]},
\end{equation}
where for each SP, $S_{\rm SP} = k L_{\rm SP}$ is the classical-like
action, with $L_{\rm SP}$  the total path length.
The term $W_{\rm SP}$ is the SP quantum amplitude, constructed
from the product of all quantum amplitudes $\sigma_{j}$ acquired
along the SP.

\begin{figure}
  \centering
  \includegraphics[width=0.9\columnwidth]{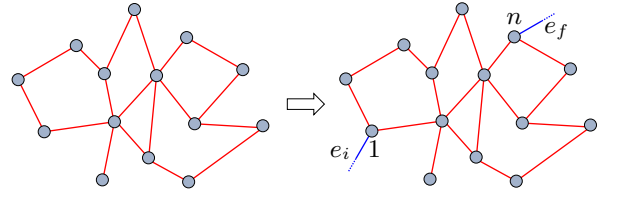}
  \caption{Graph with two leads added turning it into an open graph.}
  \label{fig:fig2}
\end{figure}

Our first goal is to rewrite the GF in a way that it is dependent on of
the underlying graph.
This will be achieved by using the adjacency matrix of the graph and the
following rules:
(i) for every vertex $j$ of the graph we define a scattering matrix
$\boldsymbol{\sigma}_{j}(k)$ associated with the BC used
at the vertex $j$,
(ii) the free propagation along the edge between two vertices $i$ and
$j$ contributes with the term $z_{ij} = z_{ji} = e^{i k \ell_{ij}}$,
where $\ell_{ij}$ is the length of the edge,
and
(iii)
in each edge between the vertices $i$ and $j$ we define two FP, one
going from $i$ to $j$ and another in the reverse direction.
They are given by
\begin{equation}
  \label{eq:pij}
  p_{ij} =
    \sum_{l \in {E_{j}^{n}}} z_{ij}
    \sigma_{j}^{[\e{j,l},\e{i,j}]}A_{jl}p_{jl}
    +\delta_{jn} z_{in}\sigma_{n}^{[e_f,\e{i,n}]},
\end{equation}
and the family $p_{ji}$ is given by the same expression above, but with
the swapping of indices $i$ and $j$.
Then, in each vertex $i$ we associated one $p_{ij}$ for every $j \in E_{i}$.
The last term in \eqref{eq:pij} is the transmission amplitude at the
vertex $n$ from the edge $\e{i,n}$ to the lead $e_f$.
So, using the above rules, the exact scattering GF for a quantum graph
with adjacency matrix $A(\Gamma)$ can be written as
\begin{equation}
  G_{\Gamma} = \frac{m}{i\hbar^{2} k}
  T_{\Gamma} e^{i k (x_{i}+x_{f})},
\end{equation}
where
$T_{\Gamma} = \sum_{j \in E_{i}} \sigma_{i}^{[\e{i,j},e_{i}]} A_{ij} p_{ij}$.
Thus we observe that, by employing the adjacency matrix of the graph, we
were able to replace an \textit{infinite} sum over SP by a
\textit{finite} sum over FP in a unique way (except for the possible
permutations of the adjacency matrix of the graph).
The number of FP is always finite.
For instance, in the fully connected simple graph on $n$ vertices,
$K_{n}$, the number of different FP is twice the number of edges,
$2\binom{n}{2}$.
We can use the Schur-Hadamard product \footnote{
Given two matrices $M$ and $N$ of the same order, the Schur-Hadamard
product is defined by $  (M \circ N)_{ij} := (M)_{ij}(N)_{ij}$.}
to know which FP need to be considered in a specific graph,
$P_{\Gamma}= P \circ A(\Gamma)$, where $P=(p_{ij})$ is an $n \times n$
matrix.
The main diagonal elements of $P_{\Gamma}$ are zero because no vertex is
connected to itself in simple graphs.
The FP altogether form a system of equations whose solution provides the
exact energy-dependent GF.
Once having obtained the exact GF, we have all the
possible information from a quantum system \cite{Book.2006.Economou}.
For instance, we can calculate the transmission probability for
transverse the graph as a function of the energy of the incident
particle, which can be used, for example, to study the presence of
resonances \cite{APPA.124.1087.2013}.
Indeed, $|T_{\Gamma}|^{2}$ represents the global transmission
probability from the lead $e_i$ to the lead $e_f$ and
it is constructed from the individual quantum amplitudes.
This kind of construction was already explored, although using a
different approach, in \cite{JPA.42.295205.2009,NPB.828.515.2010}.
Bound-state energies can be obtained from the poles of the GF and the
associated wave functions from the respective residues
\cite{PR.647.1.2016}.

The construction presented so far is for general quantum graphs.
Given the fact that star graphs can be employed as building blocks for
larger graphs \cite{Book.2012.Berkolaiko}, let us focus on the problem
of a quantum star graph, $S_n$.
Additionally, in order to simplify the notation, here and henceforth, we
drop the edge labels from the scattering amplitudes and just use $r$
($t$) for the reflection (transmission).
To prove the unitary equivalence between the SSA and the GFA for quantum
graphs, we start with the SSA by writing the general solutions for the
eigenvalue equation \eqref{eq:eigen} on the edges of the $S_{n}$:
\begin{align}
  \label{eq:psiSn}
  \psi_{\e{1,i}}(x) = {} & a_{1i} e^{i k x}+b_{1i} z_{1i} e^{-i k x},
\end{align}
$\forall i \in E_{1}$, where $a_{1i}$ and $b_{1i}$ are ($k$ dependent)
complex amplitudes (we label the central vertex as $1$).
By applying the BC \eqref{eq:bc} on the vertices of the
quantum star graph, we find
\begin{equation}
  \label{eq:evMSnS}
  U_{S_{n}}^{S}(k) \mathbf{a}_{S_{n}}= \mathbf{a}_{S_{n}},
\end{equation}
where $\mathbf{a}_{S_n}=(a_{12},b_{12},\ldots,a_{1n},b_{1n})^{T}$, and the
$2(n-1) \times 2(n-1)$ matrix $U_{S_{n}}^{S}(k)$ can be written as a
product of two matrices,
\begin{equation}
  \label{eq:MSnSSA}
  U_{S_n}^{S}(k)= S_{S_n}(k) D_{S_n}(k),
\end{equation}
with
$D_{S_n}(k) = \diag(z_{12},z_{12},z_{13},z_{13},\ldots,z_{1n},z_{1n})$
and
\begin{align}
  \label{eq:SSn}
  S_{S_{n}}(k)=
  \begin{pmatrix}
    0 & r_{2} & 0 & 0 &  \ldots & 0  & 0\\
    r_{1} & 0 & t_{1} & 0 &  \ldots & t_1  & 0\\
    0 & 0 & 0 & r_{3} &  \ldots  & 0 & 0\\
    t_{1} & 0 & r_{1} & 0 & \ldots & t_1 & 0\\
    \vdots & \vdots & \vdots & \vdots &\ddots & \vdots & \vdots\\
    0 & 0 & 0 & 0 &  \ldots & 0 & r_{n}\\
    t_1 & 0 & t_1 & 0 & \ldots & r_1 & 0
  \end{pmatrix},
\end{align}
The scattering amplitudes $r_i$, for $i \in E_{1}$, are given by
\eqref{eq:sigma} with $d_i=1$, while $r_1$ and  $t_1$ are given by
\eqref{eq:sigma} with $d_{1}=n-1$.
The edge propagation matrix $D_{S_{n}}(k)$ has the metric information of
the quantum star graph and the scattering matrix $S_{S_{n}}(k)$ has the
information of the scattering process at the vertices.
From the relations for the scattering amplitudes in
\eqref{eq:sigma_rel}, it follows that  $S_{S_{n}}(k)$ is unitary, and
the unitarity of $D_{S_{n}}(k)$ is direct.
Thus $U_{S_n}^{S}(k)$ is also unitary and it is referred to as
\textit{the quantum evolution map} \cite{AP.55.527.2006}.
The action of this map is a composition of a propagation along the edges
followed by a scattering process at the vertices.
The system \eqref{eq:evMSnS} has a nontrivial solution for the
wavenumber $k>0$, when
\begin{equation}
  \label{eq:secular_SnS}
 \zeta_{S_n}^{S}(k):=\det\left[\mathbbm{1} - U_{S_n}^{S}(k)\right] = 0,
\end{equation}
which is the \textit{secular determinant} and whose zeros define
the quantum star graph spectra.

\begin{figure}
  \centering
  \includegraphics[width=0.4\columnwidth]{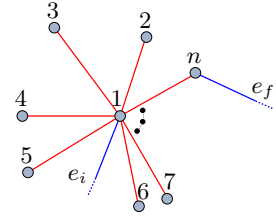}
  \caption{Star graph on $n$ vertices with two leads added turning it
    into an open star graph.}
  \label{fig:fig3}
\end{figure}

Now, consider the quantum star graph when we attach a lead $e_{i}$ to
the vertex $1$ and a lead $e_{f}$ to the vertex $n$
\footnote{In other words, we label our vertices in a such way that the
  entrance vertex is labeled $1$ and the exit vertex is labeled $n$.} as
depicted in Fig. \ref{fig:fig3}.
Thus the scattering GF is
\begin{equation}
  G_{S_{n}} =
  \frac{m}{i\hbar^2 k}
  T_{S_{n}} e^{ik(x_{i}+x_{f})},
\end{equation}
with
$T_{S_{n}}= \sum_{j \in E_{1}}t_{1}p_{1j}$, where the
families
\begin{align}
  \label{eq:system_PSn}
  p_{1j}= {} & z_{1j} (r_{1} p_{j1} + \delta_{jn} t_{n}), \nonumber\\
  p_{j1}= {} & z_{1j} (r_{1} p_{1j}
               +\sum_{i \in E_{1}^{1}} t_{1}p_{1i}),
\end{align}
form a system of $2(n-1)$ equations.
To compare with the SSA, we need to consider bound states.
This is accomplished by excluding the transmission at the vertex $n$ to
the lead $e_{f}$.
In this case, we can write \eqref{eq:system_PSn} as
\begin{equation}
  \label{eq:evSnG}
  U^{G}_{S_n}(k) \mathbf{p}_{S_n} = \mathbf{p}_{S_n},
\end{equation}
where
$\mathbf{p}_{S_n} = (p_{12},p_{21},\ldots,p_{1n},p_{n1})^{T}$ and
\eqref{eq:evSnG} has a non-trivial solution for $k>0$ if
\begin{equation}
  \label{eq:secular_SnG}
  \zeta_{S_n}^{G}(k): =
  \det\left[\mathbbm{1} - U_{S_n}^{G}(k)\right] = 0.
\end{equation}
Surprisingly, $U^{G}_{S_n}(k)$ can also be factored as a product of
$S_{S_n}(k)$ and $D_{S_n}(k)$, but in opposite order,
\begin{equation}
  U^{G}_{S_n}(k) = D_{S_{n}}(k) S_{S_n}(k).
\end{equation}
Thus it is also unitary.
In fact, $U^{G}_{S_n}(k)$ is the quantum evolution map, but now obtained
from the GFA.
The action of $U^{G}_{S_n}(k)$ is a composition of a
scattering at the vertices followed by a propagation along the
edges.
We have the following result for the eigenvalues of $U^{S}_{S_n}(k)$ and
$U^{G}_{S_n}(k)$.
\begin{theorem}
All the eigenvalues of the quantum evolution maps  $U^{S}_{S_n}(k)$ and
$U^{G}_{S_n}(k)$ are identical including the degeneracy.
\end{theorem}
\begin{proof}
Let $\mathbf{a}_{S_n}^{\lambda}$ be an eigenvector of the map
$U^{S}_{S_n}(k)$  with a nonzero eigenvalue $\lambda$,
$U^{S}_{S_n}(k) \mathbf{a}_{S_n}^{\lambda} =
\lambda  \mathbf{a}_{S_n}^{\lambda}$.
Given that $U^{S}_{S_n}(k)$ is a unitary map, all its
eigenvalues are nonzero and have modulus $1$.
Multiplying $D_{S_n}(k)$ on the left, we have
$ U^{G}_{S_n}(k) \left[D_{S_n}(k) \mathbf{a}_{S_n}^{\lambda}\right]
  = \lambda  \left[D_{S_n}(k)\mathbf{a}_{S_n}^{\lambda}\right]$.
$D_{S_n}(k)$ being unitary, $D_{S_n}(k)\mathbf{a}_{S_n}^{\lambda}$ is a
nonzero eigenvector of $U^{G}_{S_n}(k)$ with eigenvalue $\lambda$.
To complete our proof, we just reverse our reasoning for
$U^{G}_{S_n}(k)$.
The identical degeneracy of $U^{S}_{S_n}(k)$ and $U^{G }_{S_n}(k)$
is due to the fact that the unitary operators $S_{S_n}(k)$ and $D_{S_n}(k)$
preserve the orthogonality of the eigenvectors with the same eigenvalue.
\end{proof}
We then conclude that the secular determinants \eqref{eq:secular_SnS}
and \eqref{eq:secular_SnG} are equal, thus providing the same spectra.
The result above brings us to the following interesting and useful
result.
\begin{corollary}
The eigenvector $\mathbf{a}_{S_n}^{\lambda}$ with eigenvalue
$\lambda$, which are associated  with the wave function amplitudes,
Eq. \eqref{eq:psiSn}, can be obtained from the eigenvector
$\mathbf{p}_{S_n}^{\lambda}$  by
$\mathbf{a}_{S_n}=S_{S_n}(k) \mathbf{p}_{S_n}^{\lambda}$,
up to an arbitrary phase factor.
\end{corollary}
So, the wave functions for the quantum star graph can be obtained from
the GFA directly from $\mathbf{p}_{S_n}^{\lambda}$, without the need to
resort to the calculation of the residues of the GF \cite{PR.647.1.2016}.
We can now state our main result about the connection between the maps
$U^{S}_{S_n}(k)$ and $U^{G}_{S_n}(k)$.
\begin{claim}
  \label{thm:thm2}
  The quantum evolution maps $U^{S}_{S_n}(k)$ and $U^{G}_{S_n}(k)$ are
  unitarily similar.
\end{claim}
Given the properties of these maps, there are strong reasons to believe
that this claim works for every $n$.
Although a proof for a general $n$ is not known, we checked
this for $n=2,3,4,5$ by using Specht's theorem \cite{Book.2012.Horn}.
This theorem provides a necessary and sufficient condition to prove that
two matrices are unitarily similar.
A word $w(s,t)$ is any finite formal product of nonnegative powers of
$s$ and $t$,
$w(s,t) = s^{m_1}  t^{n_1}  s^{m_2}  s^{n_2} \ldots s^{m_k}  s^{n_k}$,
with $ m_1,n_1,\ldots,m_k,n_k \geq 0$.
The length of the word $w(s,t)$ is the nonnegative integer given by
the sum of all exponents in the word, $\sum_{i=1}^{k}(m_i+n_i)$.
\begin{theorem}[Specht's theorem \cite{Book.2012.Horn,LAA.519.278.2017}]
Two $n \times n$ complex matrices $A$ and $B$ are unitarily similar
if and only if
\begin{equation}
  \label{eq:specht}
  \tr w(A,A^{*}) = \tr w(B,B^{*}),
\end{equation}
for every word $w(s,t)$ in two noncommuting variables whose length
is at most \footnote{
For example, for $n=2$, three words need to be checked $w(s,t)=s$;
$s^2$ and $st$, and for $n=3$, seven words need to be checked
$w(s,t)=s$; $s^2$, $st$; $s^3$, $s^2t$; $s^2t^2$; and $s^2t^2st$.}
\begin{equation}
  n \sqrt{\frac{2n^2}{n-1}+\frac{1}{4}}+\frac{n}{2}-2.
\end{equation}
\end{theorem}

Given the fact that the GF is obtained from the solution of the system
of equations in \eqref{eq:system_PSn}, its final form has a important
contribution from the  secular determinant.
In fact, the GF for a quantum star graph on $n$ vertices is seen to be
\begin{equation}
  G_{S_n}= \frac{m}{i\hbar^2 k }\frac{1}{g_{S_n}}
  \prod_{i\in E_{1}} (g_{1i}+r_{i} t_{1} z_{1i}^{2})t_{1}t_{n}z_{1n}
  e^{ik(x_i+x_f)},
\end{equation}
where $g_{1i}=1-r_{1} r_{i} z_{1i}^2$ and $g_{S_n}$ is the secular
determinant in \eqref{eq:secular_SnG}.
So, the eigenvalues are the poles of the GF and these poles
are just the zeros of the secular determinant.
Thus the secular determinant for a quantum star graph on $n$ vertices,
with general boundary conditions, is obtained directly from Eq.
\eqref{eq:evSnG}.
Moreover, the poles have contribution from the classical periodic orbits
of the graph.
We can exemplify this with the quantum star graph $S_3$, for which
$g_{S_3}= (1-r_{1} r_{2} z_{12}^2)(1-r_{1} r_{3} z_{13}^2)
- r_{2} r_{3} t_{1}^{2} z_{12}^{2} z_{13}^{2}$, and it is possible to
see the  contribution of three periodic orbits: one confined in the edge
$\{1,2\}$, another one confined in the edge  $\{1,3\}$, and the last one
that covers the entire graph (see Fig. \ref{fig:fig4}).

\begin{figure}
  \centering
  \includegraphics*[width=0.33\columnwidth]{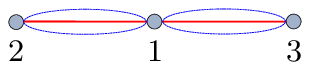}%
  \qquad
  \includegraphics*[width=0.33\columnwidth]{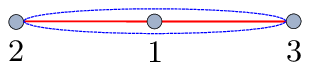}\\
  \caption{Periodic orbits for the quantum star graph $S_{3}$.}
  \label{fig:fig4}
\end{figure}

Finally, we can write a trace formula from the GFA by considering
the secular determinant \eqref{eq:secular_SnG}.
The spectral counting function $N(k)$ is given
by \cite{PRL.79.4794.1997,AoP.274.76.1999}
\begin{equation}
  N(k)= \bar{N}(k) +
  \frac{1}{\pi}\Im \sum_{\nu=1}^{\infty}\frac{1}{\nu}
  \tr\left[U_{S_n}^{G}(k)\right]^{\nu},
\end{equation}
where $\bar{N}(k)$ corresponds to the smooth part of the counting
function and the second term is the oscillatory part.
Since the main diagonal of $U_{S_n}^{G}(k)$ is zero,
$\tr\left[U_{S_n}^{G}(k)\right]^{2\nu+1}=0$ and
$\tr\left[U_{S_n}^{G}(k)\right]^{2\nu}=
2 \sum_{{p}\in\mathcal{P_{\nu}}} W_{p} e^{i k \ell_{p}}$,
where $W_{p}$ is the product of quantum amplitudes along the
periodic orbit, $\ell_{p}$ is the length of the periodic orbit, and
$\mathcal{P}_{\nu}$ is the set of periodic orbits of the graph.
For the Neumann BC, $\bar{N}= k \mathcal{L}/2\pi+1/2$
\cite{AoP.274.76.1999}, where $\mathcal{L} = 2\sum_{E} \ell_{ij}$, and
we can write the density of states $d(k)=d N(k)/dk$ as
\begin{equation}
  d(k)=
  \frac{\mathcal{L}}{2\pi} +
    \Im \frac{\ell_{p}}{\pi}
    \sum_{\nu=1}^{\infty}\frac{1}{\nu}
    \sum_{p\in\mathcal{P_{\nu}}} W_{p} e^{i k \ell_{p}}.
  \end{equation}

In summary, we have introduced a general and powerful approach for the
construction of the GF for quantum graphs based on the adjacency matrices
of the graphs.
This provides another way to obtain the secular determinant, unraveling a
unitary equivalence between the SSA and GFA.
An advantage of the GFA is that the system that leads to the secular
determinant is obtained in a very direct manner and for general energy
dependent scattering amplitudes (general BCs).
It also provides us a connection between the poles of the GF and the
secular determinant, and enables us to write a trace formula for quantum
graphs from the GFA.
Moreover, our approach can be used to study quantum walks in graphs with
complicated topologies.
This subject was studied by one of us in simple topologies in Ref. \cite{PRA.84.042343.2011}.
Furthermore, for dressed quantum graphs
\cite{PRL.88.044101.2002,PRE.65.046222.2002}, i.e., when there are
potentials $u_{ij}$ along the edges, our method can provide very good
analytical approximations for the Green's function
\cite{JPA.34.5041.2001,JPA.36.227.2002} and exact Green's function for
piecewise constant potentials \cite{PLA.378.1461.2014},
thus showing the versatility and generality of the approach developed
in this work.
These and related issues will be reported in future works
\cite{inpreparation}.

This work was partially supported by the Brazilian agencies
CNPq (Grant No. 313274/2017-7),
Funda\c{c}\~{a}o Arauc\'{a}ria (Grant No. 09/2017), and
Instituto Nacional de Ci\^{e}ncia e Tecnologia de Informa\c{c}\~{a}o
Qu\^{a}ntica (INCT-IQ).
F.M.A. thanks Dionisio Bazeia for critical reading of the manuscript and 
thanks UCL CSQ group for hospitality where part of this work was done.

%\bibliographystyle{apsrev4-1}
%\bibliography{qgl}

\end{document}